\documentclass{article}

\usepackage{amsmath}    
\usepackage{graphicx}   
\usepackage{epstopdf}
\usepackage{verbatim}   
\usepackage{color}      
\usepackage{subfigure}  
\usepackage{algorithm}
\usepackage{array}
\usepackage{algorithmic}
\usepackage{nomencl}
\usepackage{graphicx}
\usepackage{amssymb}
\usepackage{enumerate}
\usepackage{amsthm}
\usepackage{url}
\usepackage[hidelinks]{hyperref}

\usepackage{pgfplots}
\newtheorem{example}{Example}
\newtheorem{lemma}{Lemma}

\newtheorem{observation}{Observation}

\newcommand{\finish}{~\hfill$\square$}

\newcommand{\task}{\tau}
\newcommand{\taskperiod}{T}
\newcommand{\taskexecution}{C}

\newcommand{\taskset}{\Gamma}
\newcommand{\tasknumber}{n}
\newcommand{\taskindex}{i}

\newcommand{\hyperperiod}{H}

\newcommand{\delay}{D}

\newcommand{\robuststate}{S_0}
\newcommand{\failprobablestate}{S_P}
\newcommand{\failstate}{S_F}
\newcommand{\rejuvenationstate}{S_R}

\newcommand{\rejuvenatetime}{E_r}
\newcommand{\reboottime}{E_b}
\newcommand{\switchsuccess}{\rho}

\newcommand{\switchfailrate}{\lambda_0}
\newcommand{\failrate}{\lambda}
\newcommand{\rejuvenateperiod}{T_r}
\newcommand{\rejuvenateperiodrequire}{T_{0}}
\newcommand{\lifetime}{L}
\newcommand{\reliabilitiy}{R(\lifetime, \rejuvenateperiod)}
\newcommand{\availability}{A(\lifetime, \rejuvenateperiod)}
\newcommand{\reliabilityrequire}{R_{0}}
\newcommand{\mainprocessor}{\mathcal{P}_M}
\newcommand{\backupprocessor}{\mathcal{P}_B}

\newcommand{\effectivestart}{t_s}
\newcommand{\completetime}{t_f}
\newcommand{\rejuvenateperiodmin}{T_{\min}}

\newcommand{\CDF}{F}
\newcommand{\CCDF}{\overline{F}}

\newcommand{\myparagraph}[1]{\vspace{0.08in}\noindent\textbf{#1}}

\begin{document}
	
\title{Optimizing System Quality of Service through Rejuvenation for Long-Running Applications with Real-Time Constraints}

\author{Chunhui Guo\thanks{Chunhui Guo, Hao Wu, Xiayu Hua, and Shangping Ren are with the Department of Computer Science at Illinois Institute of Technology, Chicago, IL 60616, USA. Email: \{ cguo13, hwu28, xhua \}@hawk.iit.edu, ren@iit.edu}, Hao Wu, Xiayu Hua, Shangping Ren, Jerzy Nogiec\thanks{Jerzy Nogiec is with the Fermi National Accelerator Laboratory, Batavia, IL 60510, USA. Email: nogiec@fnal.gov}
}

\date{October 28, 2015}

\maketitle

\begin{abstract}
	Reliability, longevity, availability, and deadline guarantees are the four most important metrics to measure the QoS of long-running safety critical real-time applications. Software aging is one of the major factors that impact the safety of long-running real-time applications as the degraded performance and increased failure rate caused by software aging can lead to deadline missing and catastrophic consequences. Software rejuvenation is one of the most commonly used approaches to handle issues caused by software aging. In this paper, we study the optimal time when software rejuvenation shall take place so that the system's reliability, longevity, and availability are maximized, and application delays caused by software rejuvenation is minimized. In particular, we formally analyze the relationships between software rejuvenation frequency and system reliability, longevity, and availability. Based on the theoretic analysis, we develop approaches to maximizing system reliability, longevity, and availability, and use simulation to evaluate the developed approaches. In addition, we design the MIN-DELAY semi-priority-driven scheduling algorithm to minimize application delays caused by rejuvenation processes. The simulation experiments show that the developed semi-priority-driven scheduling algorithm reduces application delays by $9.01\%$ and $14.24\%$ over the earliest deadline first (EDF) and least release time (LRT) scheduling algorithms, respectively.
\end{abstract}

\section{Introduction}
\label{sec:intro}
As technology advances, computer systems become larger and more
complex --- applications are built on top of operating systems and
frameworks; they run in virtual environments and use third party
software components and services. The situation naturally makes it
more difficult or virtually impossible to develop a non-trivial system
to be completely defect-free. A class of residual software defects
produces non-catastrophic results, where applications continue to
provide their functionality, but with degraded performance or
increased use of resources. This process is typically referred to as
\textit{software aging}. 
Software aging is an accumulative process whose general characteristic
is the gradual performance degradation and/or an increase in the
software failure rate~\cite{Grottke2008ISSRE}.
As system aging progresses, the degraded performance
and accumulated errors can eventually lead to catastrophes, such
as low reliability and/or availability. For
instance, the Patriot's software failure that resulted in loss of
human life is caused by accumulated errors~\cite{or1992patriot}.  The
Mars Surveyor '98 Orbiter that launched in 1998 was designed for long
term mission to study the climate on Mars.  Unfortunately it only
worked for 83 days before it was lost in the space~\cite{stephenson1999mars}.

In addition to reliability and availability, the system longevity is another
important QoS factor for long running applications.
In Fermi National Accelerator Laboratory, there are many physics
experiments conducted on site, such as the Main Injector Neutrino
Oscillation Search (MINOS), Muon g-2, NOvA, and
muon-to-electron-conversion (Mu2e), to name a
few~\cite{fermiexp}. These experiments need to run for as long as
possible in order to observe desirable results, and at the same time,
they need also be highly reliable and available during the
experimental time. Hence, not only reliability and availability
requirements of the control system need to be met; more importantly,
the control systems that support the experiments must also be able to
run for long time period because restarting these experiments has
large human labors and financial cost. Unfortunately, aging effects
significantly impact the longevity of these control systems at
Fermilab.

Fig.~\ref{fig:motivation} illustrates a fifteen day CPU utilization and
memory usage of labVIEW~\cite{labview} running on a Fermilab machine
that monitors hundreds of sensors at Fermilab. In theory, the resource
consumption of the monitoring tool shall remain constant as it is
running on a clean server and the sampling interval and data size are
constant. However, from Fig.~\ref{fig:motivation}, we can clearly
observe that both CPU and memory consumption increase linearly with
time. Once CPU and/or memory usage level becomes too high, the system
stops working properly.

\begin{figure}[ht]
\centering
\pgfplotsset{
tick label style={font=\small},
label style={font=\small},
legend style={font=\footnotesize}
}
\begin{tikzpicture}[scale=1]
	\begin{axis}[
		xlabel=Days,
		xlabel near ticks,
		ymin=75,ymax=100,
		ylabel=CPU Utilization (\%),
		ylabel near ticks,
		legend style={at={(0.02,0.1)},anchor=west},
		no markers]
		
	\addplot [solid, thin]table[x=Hours,y=CPU] {fig/motivation.txt};
	\label{plot_one}

	\end{axis}
	
	\begin{axis}[
		hide x axis,
		axis y line*=right,
		ylabel=Memory (KB),
		ylabel near ticks,
		legend style={at={(0.42,0.13)},anchor=west},
		no markers]
		
	\addplot [dotted, thick]table[x=Hours,y=Memory] {fig/motivation.txt};
	\addlegendimage{/pgfplots/refstyle=plot_one}
	\addlegendentry{\footnotesize CPU Utilization}
	\addlegendentry{\footnotesize Memory Usage}

	\end{axis}
\end{tikzpicture}
\caption{Aging Effect on Fermi Monitoring System}
\label{fig:motivation}
\end{figure}
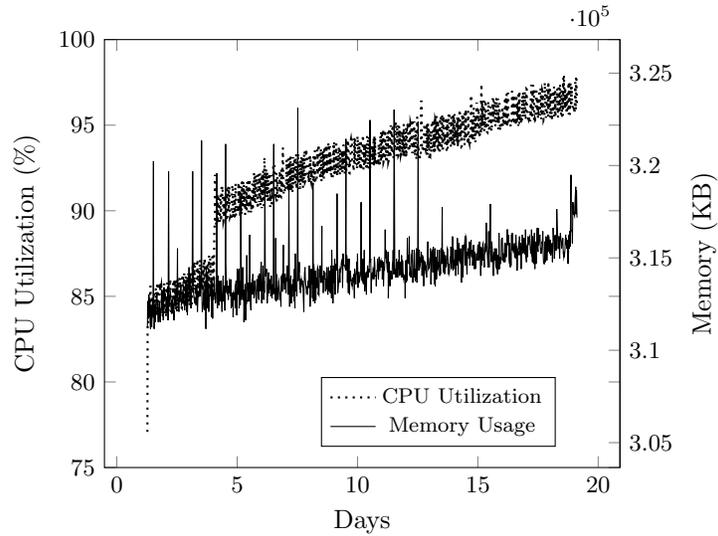

Both hardware and software aging can potentially impact system's
reliability, availability and longevity.  However, hardware wear and
tear aging often takes longer time to show effects on computer
systems~\cite{black1969electromigration}; while on the other hand,
software aging happens more frequently compared to hardware aging, and
software failures cause more outages than hardware failures in today's
computer systems~\cite{garg1998methodology}.
As software aging is inevitable~\cite{Parnas1994ICSE},
software rejuvenation is proposed as a preventive and proactive
fault-tolerance technique to deal with the aging issues~\cite{Huang1995SoftwareRejuvenation}.
Significant amount of research is devoted to
address how to perform software rejuvenation and different approaches
are proposed to rejuvenating software at different
levels. In~\cite{cotroneo2014survey}, Cotroneo et al.  surveyed over
four hundred recent research papers in the area of software aging and
rejuvenation techniques.  A comparative experimental study of software
rejuvenation overhead can be found in~\cite{alonso2013comparative}.

However, in this paper, rather than study \textit{how} to perform
software rejuvenation, we focus on \textit{when} to perform software
rejuvenation and the relationship between software rejuvenation time
points and system QoS in terms of system reliability, availability
and longevity. We consider three types of control applications:
(1) applications that need high reliability within its given lifetime,
(2) applications that need long longevity under a reliability constraint,
and (3) applications that need high availability under given reliability and longevity
constraints.  For each of these three types of applications, we present
an optimal software rejuvenation period. In addition, we develop a
semi-priority-driven MIN-DELAY scheduling algorithm that minimizes
application execution delay caused by performing system rejuvenation.

The rest of the paper is organized as follows: we discuss related work
in Section~\ref{sec:related}.  System models and assumptions the paper
is based upon are presented in Section~\ref{sec:model}. A formal
definition of the problem the paper is to address is also presented in
Section~\ref{sec:model}. 
System reliability, longevity, and availability maximizations are discussed in
Section~\ref{sec:reliability}, Section~\ref{sec:longevity}, and Section~\ref{sec:availability},
respectively. We introduce a semi-priority-driven MIN-DELAY scheduling
algorithm in Section~\ref{sec:scheduling}.  In each of the sections
where mathematical analysis is performed or a new algorithm is developed, i.e., Section~\ref{sec:reliability}, Section~\ref{sec:longevity},
Section~\ref{sec:availability}, and Section~\ref{sec:scheduling}, we
have a subsection to discuss simulation results. Finally,
Section~\ref{sec:conclusion} concludes the paper.

\section{Related Work}
\label{sec:related}
Reliability, availability and longevity are three different but
correlated factors that measure a system's QoS. A highly reliable
system often has high availability and can run for a long period of
time. Hence, researchers and engineers have been mainly focused on
system reliability issues.  Many fault-tolerance mechanisms have been
developed to improve system's reliability. A commonly used
fault-tolerance mechanism is
redundancy~\cite{Platis2008SemiMarkovAvailability2LevelRejuvenation,
  Hanmer2010FiveState}. Redundancy refers to systems that use backup
components with the same functionality as the running components. When
failures occur, systems switch the functionality to their backup
components to maintain operation continuity.  Replication is also a
widely used fault-tolerance mechanism~\cite{Singh1981TMR}. Replication
ensures computation and data are duplicated on the replicas and a
voting scheme is used to decide the correct answers of the
system. Another widely adapted fault-tolerance technique to deal with
system failures is checkpointing and
re-execution~\cite{Bobbio1999CheckpointPetriNets,Zheng2013Checkpoint}. With
checkpointing, the failed system is recovered from previously stored
correct state and re-executed only from the checkpointed state.  These
fault-tolerance techniques aforementioned may not be able to solve
software aging issues unless a failure causes the system to reboot
which resets the system to a fresh and healthy sate.

Software rejuvenation has become a commonly used preventive and
proactive maintenance approach for handling system aging.  It
is first proposed by Huang et
al.~\cite{Huang1995SoftwareRejuvenation}, and is adopted in different
domains, such as telecommunication
systems~\cite{Huang1995SoftwareRejuvenation, Hanmer2010FiveState} and
long-life deep-space mission systems~\cite{Tai1998OnBoardMaintenance,
  Tai1997OptimalDutyPeriod, Tai1998ModelAnalysis}.

Huang et al. developed a
four-state model in which a computer system operates, i.e., the
\textit{Robust} State, \textit{Failure Probable} State,
\textit{Failure} State, and \textit{Rejuvenation} State~\cite{Huang1995SoftwareRejuvenation}.  Since then,
many rejuvenation models have been developed by the research
community~\cite{Huang1995SoftwareRejuvenation,Hanmer2010FiveState}.
For instance, the five-state model~\cite{Hanmer2010FiveState} adds a new
state called \textit{Preparing} State to represent when a system
finishes executing tasks or migrating tasks to another processor if the
system has a backup component. Koutras et al. extended the initial
rejuvenation model by considering two levels of rejuvenation
actions~\cite{Platis2008SemiMarkovAvailability2LevelRejuvenation,
  Limnios2005NonparametricEstimatioReliability}, i.e., perfect
rejuvenation action and minimal rejuvenation action. The perfect
rejuvenation (cold rejuvenation) results in system returning to the
Robust State (initial state), while the minimal rejuvenation (warm
rejuvenation) results in system returning to the Failure Probable
State (the state before rejuvenation). The cost of minimal
rejuvenation is much less than the perfect rejuvenation.

To analyze software aging and study aging related failures, Trivedi et
al.~\cite{Trivedi2000SoftwareRejuvenationModelMeasure} presented two
approaches: analytic modeling approach for determining optimal times
to rejuvenate and measurement based approach for failure detection and
validation.  Tai et al.~\cite{Tai1997OptimalDutyPeriod} identified key
factors that may impact system reliability and developed an approach
to maximizing system reliability by analyzing the optimal interval
between maintenances. Okamura et
al.~\cite{Dohi2008AperiodicRejuvenationOptAvailability} discussed an
maintenance policy that combines aperiodic rejuvenation and periodic
checkpoints to maximize the system availability. The estimators of
reliability and availability were analyzed
in~\cite{Limnios2005NonparametricEstimatioReliability,
  Platis2008MaxLikelihoodEstimatorAvailabilityReliability}.

In this paper, we study \textit{when} to perform rejuvenation to
improve system's reliability, longevity, and availability for long-running 
applications with real-time constraints.  In the study, both
transient failures caused by aging effects and network transmission
failures caused by migrating applications between main and backup
processing units are taken into consideration in determining an optimal
rejuvenation period. In addition, we also study how a task scheduling
algorithm can minimize application execution delay caused by system
rejuvenation processes.

\section{System Models and Problem Formulation}
\label{sec:model}
In this section, we first introduce the models and assumptions our
work is based upon and then formulate the problem we are to address in
the paper.

\subsection{Models and Assumptions}
\label{subsec:model}

\noindent \textbf{Processing Unit State Transition Model}

We adopt the same model and assumptions used
in~\cite{Huang1995SoftwareRejuvenation}, i.e., we assume a processing unit
has the following four states, and the state transition model is shown
in Fig.~\ref{fig:state}.
\begin{itemize}
\item Robust State $\robuststate$: the processing unit starts in this state.
\item Failure Probable State $\failprobablestate$: the processing unit goes
  into this state after continuously running for some time.
\item Failure State $\failstate$: the processing unit may go into the
  failure state from the failure probable state $\failprobablestate$.
  Once the processing unit is in failure state it has to be rebooted in
  order to go back into the robust state $\robuststate$. The time it
  takes for the processing unit to reboot is $\reboottime$
\item Rejuvenation State $\rejuvenationstate$: from the failure
  probable state $\failprobablestate$ the processing unit may also go into
  the rejuvenation state $\rejuvenationstate$.  The processing unit performs
  software rejuvenation once it enters into the state and goes into
  the robust state $\robuststate$ once the rejuvenation process is
  completed. The time it takes for the processing unit to go through
  rejuvenation is $\rejuvenatetime$.
\end{itemize}
\begin{figure}[ht]
  \centering
  \includegraphics[width = 0.5\textwidth]{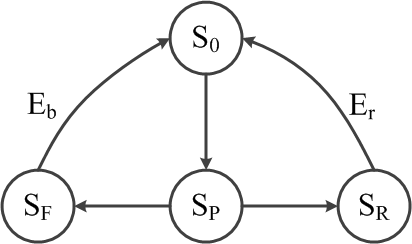} 
  \caption{processing unit State Transition Model with Rejuvenation~\cite{Huang1995SoftwareRejuvenation}}
  \label{fig:state}  
\end{figure}

The processing unit is unavailable when it goes through either reboot or
rejuvenation process. The processing unit downtime caused by each reboot or
rejuvenation is assumed to be a constant $\reboottime$ and
$\rejuvenatetime$, respectively. We assume $\reboottime \gg
\rejuvenatetime$.  Hence, our goal is to prevent the processing unit ever
enters into the failure state $\failstate$ through rejuvenation. \\

\noindent \textbf{System Model}

To guarantee timing and QoS constraints,
we use the same two-processor architecture as in~\cite{Tai1997OptimalDutyPeriod,Guo15Reliability}.
More specifically, we assume the system contains two homogeneous
and independent processing units, i.e., a main processing unit
$\mainprocessor$ and a backup processing unit $\backupprocessor$.
Though the two processing units can be dedicated to real-time applications and
alternate between being idle or going through rejuvenation and
processing real-time tasks similar to~\cite{Tai1997OptimalDutyPeriod,Guo15Reliability},
with this approach, the two processing units are not fully
utilized and in fact, most of the time, at least one processing unit is in
an idle state. To better utilize both of the processing units and avoid
resource waste, we assume that the main processing unit $\mainprocessor$
executes real-time tasks and the backup processing unit $\backupprocessor$
executes non real-time tasks when the main processing unit operates
correctly, and executes real-time tasks when $\mainprocessor$ goes
through maintenance mode. The system model is shown in
Fig.~\ref{fig:model}.

\begin{figure}[ht]
  \centering
  \includegraphics[width = 0.7\textwidth]{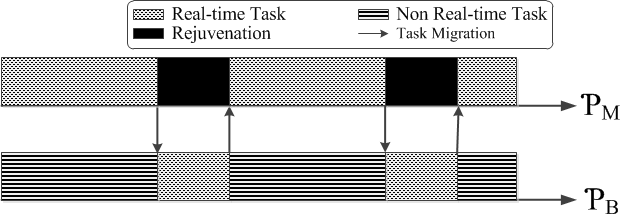} 
  \caption{System Model}
  \label{fig:model}  
\end{figure}

To avoid failure caused by aging effects, we assume that the main
processing unit does rejuvenation periodically with a period
$\rejuvenateperiod$.  Rejuvenation takes $\rejuvenatetime$ to
complete, we assume $\rejuvenateperiod > \rejuvenatetime$.  When the
main processing unit starts a rejuvenation process, to guarantee real-time
tasks still meet their deadlines, some or all of the real-time tasks
on the main processing unit are migrated to the backup processing unit to continue
their executions~\cite{Tai1997OptimalDutyPeriod}. The backup processing unit temporally suspends its non
real-time tasks and gives higher priority to the real-time tasks
migrated from the main processing unit. For planned rejuvenation, the start
time of a rejuvenation process is known a priori, hence we can
reasonably assume that, from real-time task's perspective, the
overhead for backup processing unit to suspend its execution of non
real-time tasks is negligible~\cite{Tai1997OptimalDutyPeriod}.

Furthermore, as the backup processing unit does not contain tasks with
deadline constraints, it can frequently be rebooted or rejuvenated at
time when the main processing unit operates in robust state.  Hence, we can
further assume that the backup processing unit $\backupprocessor$ is always
in the robust state when the main processing unit $\mainprocessor$ is in
rejuvenation state.
\\

\noindent \textbf{Real-Time Task Model}

The real-time task model considered in this paper is similar to the
one defined by Liu and Layland~\cite{CLLiu1973}. A task set
$\taskset=\{\task_1, \task_2, \dots, \task_\tasknumber\}$ has
$\tasknumber$ independent periodic tasks that are all released at time
0.  Each task $\task_\taskindex \in \taskset$ is a 2-tuple
$(\taskperiod_\taskindex, \taskexecution_\taskindex)$, where
$\taskperiod_\taskindex$ is the \textit{inter-arrival time} between
any two consecutive jobs of $\task_\taskindex$ (also called
\textit{period}), and $\taskexecution_\taskindex$ is the
\textit{worst-case execution time (WCET)}.  The deadline of each task is
equal to its period. The hyper-period of $\taskset$ is defined as the
LCM (Least Common Multiple) of each task's period, i.e., $\hyperperiod
= \mathtt{LCM} \{ \taskperiod_1, \taskperiod_2, \cdots,
\taskperiod_\tasknumber \}$.

In the system, task preemptions and task migrations are permitted. We also
assume that the overhead associated with task preemptions and 
migrations is considered into the task's worst case execution time.
\\

\noindent \textbf{Network Failure Model}

As the main processing unit $\mainprocessor$ and the backup
processing unit $\backupprocessor$ may locate on different
computers, the task migrations between the two processing units
need to be completed over a network.
We take the same assumption as in~\cite{Bertsekas1992DataNetwork} that 
the network transmission failure model follows Poisson
distribution, i.e., it has a constant failure rate $\switchfailrate$.
Task migrations between $\mainprocessor$ and $\backupprocessor$ may
fail because of network transmission failures. With constant network
transmission failure rate, the probability of a successful task
migration is hence a constant and it is denoted as $\switchsuccess$.
If $\mainprocessor$ and $\backupprocessor$ locate on the same
computer, then $\switchsuccess = 1$.
\\

\noindent \textbf{Aging Caused Transient Failure Model}

Since transient faults are more frequent than permanent
faults~\cite{Murray2003TransientMorePermanent}, we only consider the
transient faults.  As the system deteriorates with aging, we assume
that the transient failure rate $\failrate(t)$ increases with time
$t$~\cite{Grottke2008ISSRE,tobias2011book}. The CDF (Cumulative Distribution Function) of transient fault is
modeled as $\CDF(t) = 1-e^{-\int_{0}^{t} \failrate(x)
  dx}$~\cite{Barlow1996Reliability}.
 
After each rejuvenation, the system transient failure rate and
cumulative distribution function are reset to
$\failrate(\completetime) = \failrate(0) = 0$ and $\CDF(\completetime)
= \CDF(0) = 0$, where $\completetime$ is the time point when a
rejuvenation process completes.

Fig.~\ref{fig:reliability} illustrates the behaviors of system
rejuvenation and transient failure rate.
\begin{figure}[ht]
	\centering
	\includegraphics[width = 0.7\textwidth]{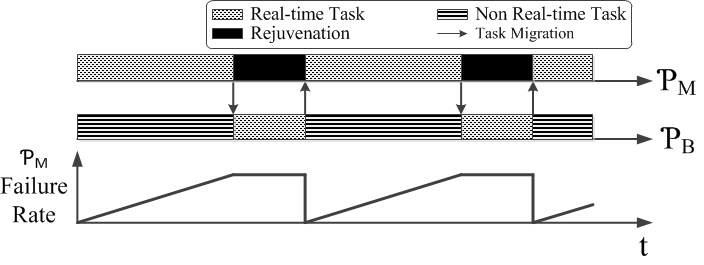} 
	\caption{System Rejuvenation and Transient Failure Rate}
	\label{fig:reliability}  
\end{figure}
\subsection{Problem Formulation}
\label{subsec:formulation}

The models and assumptions defined in Section~\ref{subsec:model}
indicate that the system reliability decreases over time because of
the increased failure rate caused by aging effects. To maintain system
reliability at the required level, on one hand, the system should
perform rejuvenation frequently, but on the other hand, every
rejuvenation requires tasks being migrated to and back from the backup
processing unit. Due to unreliable network, frequent migration between
processing units can negatively affect the system reliability. Hence, there
is a balanced point as to how frequently the system shall perform
rejuvenation so that the system reliability can be maximized.
When the system reliability requirement is given, the system longevity
and availability are also impacted by rejuvenation frequency.

Furthermore, although rejuvenation can slow down aging process, i.e.,
slow down system transient failure increase rate, and improves system
reliability, each rejuvenation not only causes the main processing unit
being unavailable to process real-time tasks, it also delays the
execution of non real-time applications deployed on the backup
processing unit. In other words, system reliability, longevity,
availability, and processing delay of non real-time applications can
all be affected by the frequencies of software rejuvenation
processes. In this paper, we are to address how to maximize system
reliability, longevity, and availability and minimize delays for
long-running applications with real-time constraints.

More specifically, we consider a real-time periodic task set
$\taskset$ which is deployed on main processing unit ($\mainprocessor$), and
a backup processing unit ($\backupprocessor$) which is connected with the
main processing unit through a network.  Assume that the main processing unit
transient failure rate is $\failrate(t)$ which increases with time,
and the network transmission failure rate is a constant $\switchfailrate$
which means the probability of successful task migration is also a constant $\switchsuccess$,
we are to address following four questions:

\myparagraph{Problem 1: (Reliability Maximization)} \textit{Given a system
	longevity $\lifetime$, determine an optimal
	rejuvenation period $\rejuvenateperiod$ that maximizes the system
	reliability $\reliabilitiy$ within its operational interval
	$[0, \lifetime]$.}

\myparagraph{Problem 2: (Longevity Maximization)} \textit{Given a system
  reliability constraint $\reliabilityrequire$, determine an optimal
  rejuvenation period $\rejuvenateperiod$ that maximizes the system
  operational interval $[0, \lifetime]$ in which the system
  reliability is guaranteed at $\reliabilityrequire$.}

\myparagraph{Problem 3: (Availability Maximization)} \textit{ Given
  system reliability $\reliabilityrequire$ and longevity $\lifetime$
  constraints, determine an optimal rejuvenation period
  $\rejuvenateperiod$ that maximizes the main processing unit's availability
  $\availability$ within its operational interval $[0, \lifetime]$.  }

\myparagraph{Problem 4: (Delay Minimization)} \textit{ Given system
  reliability $\reliabilityrequire$ and longevity $\lifetime$
  constraints, and rejuvenation period $\rejuvenateperiod$, design a
  real-time task scheduling algorithm for the main processing unit that
  minimizes the delay of non real-time tasks on the backup processing unit.
}

\section{System Reliability Maximization}
\label{sec:reliability}
\subsection{Reliability Analysis}

System reliability is defined as the probability that the system
operates without failure within a given time
interval~\cite{Barlow1996Reliability}. We refer system longevity as
its longest operational interval with guaranteed reliability.  Assume
the time interval the system operates is $[0, \lifetime]$, and the
rejuvenation period is $\rejuvenateperiod$, then the system performs
$(\lceil \lifetime/ \rejuvenateperiod \rceil-1)$ times rejuvenation,
and tasks migrate $2(\lceil \lifetime/ \rejuvenateperiod \rceil-1)$
times between the main and the backup processing units.
Hence, the system reliability within its longevity interval
$[0, \lifetime]$ is

\begin{align}
\label{eq:realReliability}
\reliabilitiy = \switchsuccess^{2(\left\lceil
		\frac{\lifetime}{\rejuvenateperiod} \right\rceil-1)} \cdot 
		\CCDF(\rejuvenateperiod)^{\left\lceil
		\frac{\lifetime}{\rejuvenateperiod} \right\rceil -1}
		\cdot \CCDF(t')
\end{align}
where $t'=\lifetime - \rejuvenateperiod \cdot (\left\lceil
\lifetime/\rejuvenateperiod \right\rceil -1)$ and 
$\CCDF(\rejuvenateperiod) = 1-\CDF(\rejuvenateperiod) =
e^{-\int_{0}^{\rejuvenateperiod} \failrate(t) dt}$~\cite{Guo15Reliability}.

The following lemma gives the worst case system reliability under the
settings defined above.

\begin{lemma}
\label{lm:minReliability}
Let system longevity be $\lifetime$ and rejuvenation period be $\rejuvenateperiod$, if
$\lifetime \ \mathtt{mod} \ \rejuvenateperiod = 0$, then the system has the lowest reliability given by Eq.~\eqref{eq:reliability}
\begin{align}
	\label{eq:reliability}
	\reliabilitiy = \switchsuccess^{2(
		\frac{L}{\rejuvenateperiod} -1)} \cdot 
	\CCDF(\rejuvenateperiod)^{\frac{L}{\rejuvenateperiod}}
\end{align}
\finish
\end{lemma}

\begin{proof}	
As $\lifetime$ and
$\rejuvenateperiod$ are given, the first two factors
in Eq.~\eqref{eq:realReliability} are fixed. Hencc, the reliability
is minimal when $\CCDF(t')$ is minimal.

As $\CCDF(t)$ decreases with $t$, $\CCDF(t')$ is minimal
when $t' = \rejuvenateperiod$, i.e.,
$\lifetime \ \mathtt{mod} \ \rejuvenateperiod = 0$.
Hence, we have Eq.~\eqref{eq:reliability}.
\end{proof}

In the following discussions on system reliability, longevity, availability, and
non real-time application execution delays, we focus on the case where
the system has the worst case reliability, i.e.,
Eq.~(\ref{eq:reliability}).

\subsection{Reliability Maximization}

Based on Eq.~\eqref{eq:reliability}, system reliability is a function of two variables, i.e.,  $\lifetime$ and $\rejuvenateperiod$.
To identify the relationship between reliability and rejuvenation
period, we derive the partial derivative of $\reliabilitiy$
with respect to the variable $\rejuvenateperiod$ as follows.
\begin{align}
\nonumber
\begin{array}{rcl}
\frac{\partial \reliabilitiy}{\partial \rejuvenateperiod}
&=& -\frac{2\lifetime}{\rejuvenateperiod^2} \cdot 
	\switchsuccess^{2(\frac{\lifetime}{\rejuvenateperiod} -1)}
	\cdot \CCDF(\rejuvenateperiod)^{\frac{\lifetime}{\rejuvenateperiod}}
	\cdot \ln \switchsuccess\\
&+& \switchsuccess^{2(\frac{\lifetime}{\rejuvenateperiod} -1)} \cdot
	\CCDF(\rejuvenateperiod)^{\frac{\lifetime}{\rejuvenateperiod}}
	\cdot	(-\frac{\lifetime}{\rejuvenateperiod^2} \cdot \ln \CCDF(\rejuvenateperiod)\\
&+& \frac{\lifetime}{\rejuvenateperiod \CCDF(\rejuvenateperiod)} \cdot
	\frac{d \CCDF(\rejuvenateperiod)}{d \rejuvenateperiod})
\end{array}
\end{align}

Let $\frac{\partial R}{\partial \rejuvenateperiod} (\lifetime, \rejuvenateperiod) = 0$,
we have
\begin{align}
\label{eq:zeroPoint}
\frac{\rejuvenateperiod}{\CCDF(\rejuvenateperiod)}
\cdot \frac{d \CCDF(\rejuvenateperiod)}{d \rejuvenateperiod}
-\ln \CCDF(\rejuvenateperiod)
-2\ln \switchsuccess = 0.
\end{align}

As Eq.~\eqref{eq:reliability} is a concave function, the optimal rejuvenation period that maximizes the system reliability
can be calculated by solving Eq.~(\ref{eq:zeroPoint})
with given $\failrate(t)$ and $\switchsuccess$.

\begin{lemma}
\label{lm:optimalPeriod}
The optimal rejuvenation period is only influenced by network
transmission failure rate $\switchfailrate$ and transient
fault occurrence rate $\failrate(t)$, but not by system longevity $\lifetime$.
\finish
\end{lemma}

\begin{proof}
The lemma can be directly proven by Eq.~\eqref{eq:zeroPoint}, where
$\CCDF(t) = e^{-\int_{0}^{t} \failrate(x) dx}$, and $\switchsuccess$ is
a constant with fixed $\switchfailrate$.
\end{proof}

The Weibull distribution is commonly used to model the distribution
of transient faults~\cite{Barlow1996Reliability}, with failure rate
$\failrate(t) = k t^{k-1} / r^k$
and cumulative distribution function $\CDF(t) = 1- e^{-(t/r)^k}$, 
where $r > 0$ and $k>0$ are scale and shape parameters.
The failure rate increases with time $t$ if $k>1$.

In Section~\ref{subsec:model}, we have made the assumption that due to aging effects, the system  transient failure rate increases with time.  Hence, we can use Weibull distribution with $k >1$ to model aging effects.  Substitute $\CCDF(t) = e^{-(t/r)^k}$ into
Eq.~(\ref{eq:zeroPoint}) and solve the equation,
we obtain the optimal rejuvenation period
that maximizes the system reliability as follows
\begin{align}
\label{eq:period}
\rejuvenateperiod^* = \sqrt[k]{\frac{2 r^k \ln \switchsuccess}{1-k}}.
\end{align}

\subsection{Simulation Results}

We use simulation to evaluate the relationship
between rejuvenation period and system reliability. 
The simulation parameters are set as following:
\begin{itemize}
\item Failure rate: $\failrate(t) = 3 t^2 / 10^9$
\item Probability of a successful task
	migration between $\mainprocessor$ and $\backupprocessor$:
	$\switchsuccess = 0.99999$
\item System Longevity: $\lifetime \in \{ 100, 1000 \}$
\item Rejuvenation periods: $\rejuvenateperiod \in \{1,5,10,\dots,95,100\}$
\end{itemize}

For each rejuvenation period, we use Eq.~\eqref{eq:realReliability}to calculate the system reliability $\reliabilitiy$. Fig.~\ref{fig:reliabilityExp}
shows the system reliability under different rejuvenation
periods for both longevity settings.

\begin{figure}[ht]
	\centering
	\pgfplotsset{
tick label style={font=\scriptsize},
label style={font=\scriptsize},
legend style={font=\scriptsize},
every axis/.append style={
thick,
tick style={semithick}}
}
	
	\begin{tikzpicture}
		\begin{axis}[
		xlabel=\text{\small{$\rejuvenateperiod$}},
		xmin=0,xmax=100,
		xtick={0,10,20,30,40,50,60,70,80,90,100},
		xlabel near ticks,
		ymin=0.98,ymax=1.0,	
		scaled y ticks=base 10:2,			
		ylabel=\text{\small{$\reliabilitiy$}},
		ylabel near ticks,
		legend style={at={(0.66,0.12)},anchor=west}		
		]

		\addplot [solid, mark=star]table[x=Ts,y=R] {fig/reliability100.txt};
		\addplot [dotted, mark=star]table[x=Ts,y=R] {fig/reliability1000.txt};

		\legend{$\lifetime = 100$, $\lifetime = 1000$};
		\end{axis}
	\end{tikzpicture}
	\caption{Reliability vs Rejuvenation Period}
	\label{fig:reliabilityExp}
\end{figure}
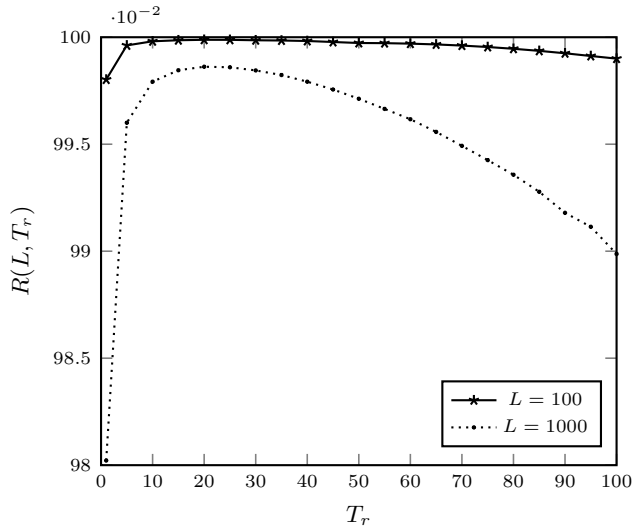

From Fig.~\ref{fig:reliabilityExp}, we have the
following observations:
\begin{enumerate}
\item When the rejuvenation period increases, the system
	reliability first increases and then decreases.
\item Neither too small rejuvenation period nor too large
	rejuvenation period has positive impact on the system
	reliability. Too frequent rejuvenation in fact lowers system reliability.
\item The experimental optimal rejuvenation period that maximizes
	the system reliability is consistent with the mathematical
	analysis (Eq.~\eqref{eq:zeroPoint}). In particular, for
	both experiment settings the optimal rejuvenation period
	$\rejuvenateperiod = 20$ is the nearest value with the
	mathematical analysis result $\rejuvenateperiod^* = 21.54$
	(Eq.~\eqref{eq:period}) among all provided rejuvenation
	periods.
\item The system
	longevity, i.e., its operation time, does not impact the optimal
	rejuvenation period for maximizing system reliability 
	which is consistent with Lemma~\ref{lm:optimalPeriod}.
	In particular, the optimal rejuvenation is $\rejuvenateperiod = 20$
	for both $\lifetime = 100$ and $\lifetime = 1000$ experiment
	settings.
\end{enumerate}

\section{System Longevity Maximization}
\label{sec:longevity}
\subsection{Longevity Analysis}
System longevity is defined as the system operational time with
guaranteed reliability $\reliabilityrequire$.  Due to system aging,
the system reliability decreases when its operational time increases.
If the rejuvenation period $\rejuvenateperiod$ is given, there is a
maximal longevity with which the system reliability requirement
$\reliabilityrequire$ is guaranteed. Our question is to determine the
optimal rejuvenation period that maximizes the system's longevity
without compromising the system's reliability requirement.
Reliability decreasing rate is one of the critical factors that
impacts the system's longevity.  The slower  the reliability
function decreases, the longer the system runs reliably. 
The first derivative of the reliability function Eq.~\eqref{eq:reliability}
with respect to $\rejuvenateperiod$ represents how the reliability
$\reliabilitiy$ changes with $\rejuvenateperiod$. While the second
derivative measures how fast the reliability $\reliabilitiy$ changes
with $\rejuvenateperiod$, i.e., the reliability decreasing rate.
The reliability decreasing rate is given below
\begin{align}
\label{eq:reliabilityDecreaseRate}
\begin{array}{rcl}
  S(\rejuvenateperiod)
  &=& \frac{\partial^2 \reliabilitiy}{\partial^2 \rejuvenateperiod}
  = \frac{\partial A}{\partial \rejuvenateperiod}\\
  &=& -\frac{2\lifetime}{\rejuvenateperiod^2} \cdot \ln \switchsuccess
  \cdot A + 
  \frac{4\lifetime}{\rejuvenateperiod^3} \cdot \ln \switchsuccess
  \cdot \switchsuccess^{2(\frac{\lifetime}{\rejuvenateperiod} -1)}	
  \cdot \CCDF(\rejuvenateperiod)^{\frac{\lifetime}{\rejuvenateperiod}}\\
  &+& A \cdot (-\frac{1}{\rejuvenateperiod^2} \cdot \ln
  \CCDF(\rejuvenateperiod) 
  + \frac{\lifetime}{\rejuvenateperiod \CCDF(\rejuvenateperiod)} \cdot
  \frac{d \CCDF(\rejuvenateperiod)}{d \rejuvenateperiod})\\
  &+& \switchsuccess^{2(\frac{\lifetime}{\rejuvenateperiod} -1)}	
  \cdot \CCDF(\rejuvenateperiod)^{\frac{\lifetime}{\rejuvenateperiod}}
  \cdot ( \frac{2}{\rejuvenateperiod^3} \cdot \ln \CCDF(\rejuvenateperiod)\\
  &-& \frac{1}{\rejuvenateperiod^2} \cdot
  \frac{d \ln \CCDF(\rejuvenateperiod)}{d \rejuvenateperiod}
  + \frac{-L \cdot (\CCDF(\rejuvenateperiod) + \rejuvenateperiod
    \cdot \frac{d \CCDF(\rejuvenateperiod)}{d \rejuvenateperiod})}
  {(\rejuvenateperiod \CCDF(\rejuvenateperiod))^2} \cdot
  \frac{d \CCDF(\rejuvenateperiod)}{d \rejuvenateperiod}\\
  &+& \frac{\lifetime}{\rejuvenateperiod \CCDF(\rejuvenateperiod)} \cdot
  \frac{d^2 \CCDF(\rejuvenateperiod)}{d^2 \rejuvenateperiod} )
\end{array}
\end{align}
where $A = -\frac{2\lifetime}{\rejuvenateperiod^2} \cdot
\switchsuccess^{2(\frac{\lifetime}{\rejuvenateperiod} -1)} \cdot
\CCDF(\rejuvenateperiod)^{\frac{\lifetime}{\rejuvenateperiod}} \cdot
\ln \switchsuccess +
\switchsuccess^{2(\frac{\lifetime}{\rejuvenateperiod} -1)} \cdot
\CCDF(\rejuvenateperiod)^{\frac{\lifetime}{\rejuvenateperiod}} \cdot
(-\frac{1}{\rejuvenateperiod^2} \cdot \ln \CCDF(\rejuvenateperiod) +
\frac{\lifetime}{\rejuvenateperiod \CCDF(\rejuvenateperiod)} \cdot
\frac{d \CCDF(\rejuvenateperiod)}{d \rejuvenateperiod})$
is the first derivative of $\reliabilitiy$ with respect to $\rejuvenateperiod$.
Noting that $S(\rejuvenateperiod)$ is a concave function.

The problem of maximizing system longevity is now transformed to
determine $\rejuvenateperiod$ that minimizes the value of
$S(\rejuvenateperiod)$ given by
Eq.~\eqref{eq:reliabilityDecreaseRate}.

To solve the problem, we obtain the first derivative of
$S(\rejuvenateperiod)$ and let the derivative be zero to calculate the
optimal $\rejuvenateperiod$. The first derivative of
$S(\rejuvenateperiod)$ is calculated as follows.
\begin{align}
\label{eq:rateDerivative}
\begin{array}{rcl}
  \frac{d S(\rejuvenateperiod)}{d \rejuvenateperiod}
  &=& \frac{4\lifetime}{\rejuvenateperiod^3} \cdot \ln \switchsuccess
  \cdot A - \frac{2\lifetime}{\rejuvenateperiod^2}
  \cdot \ln \switchsuccess \cdot \frac{d A}{d \rejuvenateperiod}\\
  &-& \frac{12\lifetime}{\rejuvenateperiod^4} \cdot \ln \switchsuccess
  \cdot \switchsuccess^{2(\frac{\lifetime}{\rejuvenateperiod} -1)}	
  \cdot \CCDF(\rejuvenateperiod)^{\frac{\lifetime}{\rejuvenateperiod}}
  + \frac{4\lifetime}{\rejuvenateperiod^3} \cdot \ln \switchsuccess
  \cdot A\\
  &+& \frac{d A}{d \rejuvenateperiod}
  \cdot (-\frac{1}{\rejuvenateperiod^2} \cdot \ln \CCDF(\rejuvenateperiod)
  + \frac{\lifetime}{\rejuvenateperiod \CCDF(\rejuvenateperiod)} \cdot
  \frac{d \CCDF(\rejuvenateperiod)}{d \rejuvenateperiod})\\
  &+& 2A \cdot ( \frac{2}{\rejuvenateperiod^3} \cdot \ln
  \CCDF(\rejuvenateperiod) 
  -\frac{1}{\rejuvenateperiod^2} \cdot
  \frac{d \ln \CCDF(\rejuvenateperiod)}{d \rejuvenateperiod}\\
  &+& \frac{-L \cdot (\CCDF(\rejuvenateperiod) + \rejuvenateperiod
    \cdot \frac{d \CCDF(\rejuvenateperiod)}{d \rejuvenateperiod})}
  {(\rejuvenateperiod \CCDF(\rejuvenateperiod))^2} \cdot
  \frac{d \CCDF(\rejuvenateperiod)}{d \rejuvenateperiod}
  + \frac{\lifetime}{\rejuvenateperiod \CCDF(\rejuvenateperiod)} \cdot
  \frac{d^2 \CCDF(\rejuvenateperiod)}{d^2 \rejuvenateperiod} )\\
  &+& \switchsuccess^{2(\frac{\lifetime}{\rejuvenateperiod} -1)}	
  \cdot \CCDF(\rejuvenateperiod)^{\frac{\lifetime}{\rejuvenateperiod}}
  \cdot ( \frac{-6}{\rejuvenateperiod^4} \cdot \ln \CCDF(\rejuvenateperiod)\\
  &+& \frac{2}{\rejuvenateperiod^3} \cdot
  \frac{d\ln \CCDF(\rejuvenateperiod)}{d \rejuvenateperiod}
  + \frac{2}{\rejuvenateperiod^3} \cdot
  \frac{d\ln \CCDF(\rejuvenateperiod)}{d \rejuvenateperiod}\\
  &-& \frac{1}{\rejuvenateperiod^2} \cdot
  \frac{d^2 \ln \CCDF(\rejuvenateperiod)}{d^2 \rejuvenateperiod}
  + \frac{d B}{d \rejuvenateperiod} \cdot
  \frac{d \CCDF(\rejuvenateperiod)}{d \rejuvenateperiod}
  + B \cdot \frac{d^2 \CCDF(\rejuvenateperiod)}{d^2 \rejuvenateperiod}\\
  &+& B \cdot \frac{d^2 \CCDF(\rejuvenateperiod)}{d^2 \rejuvenateperiod}
  + \frac{\lifetime}{\rejuvenateperiod \CCDF(\rejuvenateperiod)} \cdot
  \frac{d^3 \CCDF(\rejuvenateperiod)}{d^3 \rejuvenateperiod})
\end{array}
\end{align} 
where $B=\frac{-L \cdot (\CCDF(\rejuvenateperiod) + \rejuvenateperiod
  \cdot \frac{d \CCDF(\rejuvenateperiod)}{d \rejuvenateperiod})}
{(\rejuvenateperiod \CCDF(\rejuvenateperiod))^2}$ and $\frac{d B}{d
  \rejuvenateperiod} = -\lifetime \cdot \frac{ (2\frac{d
    \CCDF(\rejuvenateperiod)}{d \rejuvenateperiod}) \cdot
  (\rejuvenateperiod \CCDF(\rejuvenateperiod))^2 -
  2(\CCDF(\rejuvenateperiod) + \rejuvenateperiod \frac{d
    \CCDF(\rejuvenateperiod)}{d \rejuvenateperiod})^2}
{(\rejuvenateperiod \CCDF(\rejuvenateperiod))^4}$.

As $S(\rejuvenateperiod)$, i.e. Eq.~\eqref{eq:reliabilityDecreaseRate}, is a concave function,
the value $\rejuvenateperiod^*$ that satisfies $\frac{d
  S(\rejuvenateperiod)}{d \rejuvenateperiod} = 0$ is the optimal
rejuvenation period that maximizes the system longevity under
reliability requirement $\reliabilityrequire$.

\subsection{Simulation Results}

We use simulation to evaluate the relationship between rejuvenation
period and system longevity under a given reliability constraint. The
simulation parameters are set as following:
\begin{itemize}
\item Failure rate: $\failrate(t) = 3 t^2 / 10^9$
\item Probability of a successful task
migration between $\mainprocessor$ and $\backupprocessor$:
$\switchsuccess \in \{0.99999, 0.999999\}$
\item Reliability requirement: $\reliabilityrequire = 0.9997$
\item Rejuvenation periods: $\rejuvenateperiod \in \{1,5,10,\dots,95,100\}$
\end{itemize}

As the system reliability decreases when the rejuvenation period
increases, we calculate the maximal longevity satisfying
$\reliabilityrequire$ by increasing the longevity from $0$ until
$\reliabilityrequire$ fails.  Fig.~\ref{fig:lifetimeExp}
shows the maximal longevity that satisfies $\reliabilityrequire$ under
different rejuvenation periods.

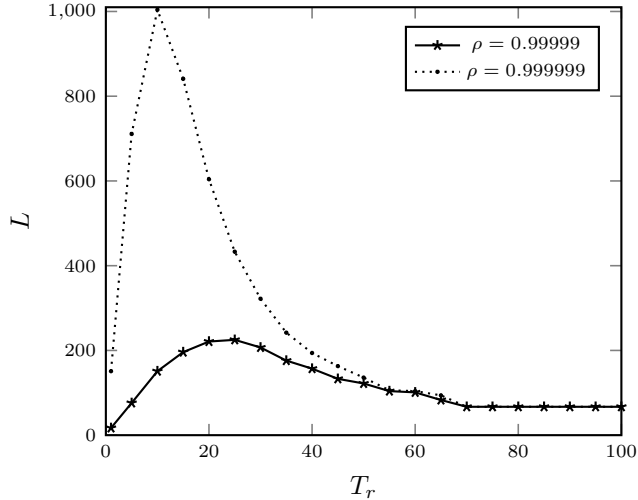
\begin{figure}[ht]
	\centering
	\pgfplotsset{
tick label style={font=\scriptsize},
label style={font=\scriptsize},
legend style={font=\scriptsize},
every axis/.append style={
thick,
tick style={semithick}}
}
	
	\begin{tikzpicture}
		\begin{axis}[
			xlabel=\text{$\rejuvenateperiod$},
			xmin=0,xmax=100,
		   	xlabel near ticks,
			ymin=0,ymax=1010,
			ylabel=\text{$\lifetime$},
			ylabel near ticks,
			legend style={at={(0.58,0.88)},anchor=west}		
			]

		\addplot [solid, mark=star]table[x=Ts,y=L] {fig/lifetime.txt};
		\addplot [dotted, mark=star]table[x=Ts,y=L2] {fig/lifetime.txt};
		\legend{$\switchsuccess = 0.99999$, $\switchsuccess = 0.999999$};
		\end{axis}
	\end{tikzpicture}
	\caption{Longevity vs Rejuvenation Period}
	\label{fig:lifetimeExp}
\end{figure}

From Fig.~\ref{fig:lifetimeExp}, we
have the following observations:
\begin{enumerate}
\item When the rejuvenation period increases, the longevity first
  increases and then decreases or remains the same. For instance, when
  $\switchsuccess = 0.99999$, the longevity increases when
  $\rejuvenateperiod$ increases from $1$ to $25$ and starts to
  decrease when $\rejuvenateperiod$ increases from $25$ to $70$.  For
  a more reliable network, i.e., when $\switchsuccess = 0.999999$, the
  system longevity reaches its maximal value of 1,004 when the
  rejuvenation period is $\rejuvenateperiod = 10$.
\item When the rejuvenation period is too short or too long, system's
  longevity is short.  For instance, $\switchsuccess = 0.99999$ and
  rejuvenation period is 1, the system longevity reaches its minimal
  value of $17$. For a more reliable network, the minimal system
  longevity is $151$.  Although it is longer compared with a less
  reliable networked system, it is much shorter than its optimal
  longevity, which is 1,004 in this case, as shown in
  Fig.~\ref{fig:lifetimeExp}.
\end{enumerate}

The reason behind these observations is that each rejuvenation
requires task migrations between the main and backup processing units.  Due
to possible failures during task migrations, more frequent
rejuvenation, i.e., a short rejuvenation period, encounters more
transmission failures and hence results in short system longevity.
When the network is more reliable, more frequent rejuvenation benefits
system longevity as shown depicted in Fig.~\ref{fig:lifetimeExp}.  On
the other hand, due to software aging, transient failures increase as
rejuvenation period increases.  Hence, when the rejuvenation period is
too long, the system longevity also decreases.

\section{System Availability Maximization}
\label{sec:availability}
\subsection{Availability Analysis}
Based on the system state model defined in Section~\ref{sec:model},
the system availability is defined as the probability that the system
is in either robust state ($\robuststate$) or failure probable state
($\failprobablestate$) at a time instant~\cite{Barlow1996Reliability}.
In essence, system availability is the ratio between the system
execution time and its longevity.

Assume within the system's longevity $\lifetime$, the main processing unit
$\mainprocessor$ performs $(\lceil \lifetime/ \rejuvenateperiod
\rceil-1)$ times of rejuvenations, the system downtime for each
rejuvenation is $\rejuvenatetime$, then the availability of the main
processing unit $\mainprocessor$ is
\begin{align}
\label{eq:availability}
\availability = \frac{\lifetime-(\left\lceil
	\frac{\lifetime}{\rejuvenateperiod} \right\rceil -1)
	 \cdot \rejuvenatetime}{\lifetime}
\end{align}
where the rejuvenation cost $\rejuvenatetime$ is a constant.

Fig.~\ref{fig:availabilityExp} plots Eq.~\eqref{eq:availability} under
the same setting as for Fig.~\ref{fig:lifetimeExp} with
$\rejuvenatetime=0.5$

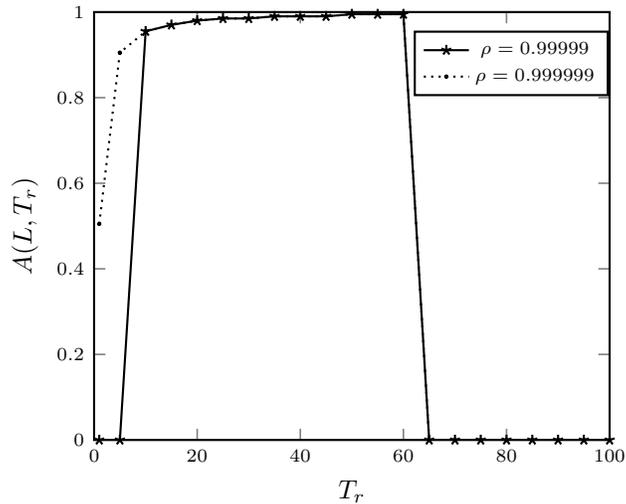
\begin{figure}[ht]
  \centering
  \pgfplotsset{
tick label style={font=\scriptsize},
label style={font=\scriptsize},
legend style={font=\scriptsize},
every axis/.append style={
thick,
tick style={semithick}}
}
	
	\begin{tikzpicture}
		\begin{axis}[
			xlabel=\text{$\rejuvenateperiod$},
			xmin=0,xmax=100,
			xlabel near ticks,
			ymin=0,ymax=1.0,
			ylabel=\text{$\availability$},
			ylabel near ticks,
			legend style={at={(0.62,0.88)},anchor=west}				
			]
			
		\addplot [solid, mark=star]table[x=Ts,y=A] {fig/availability.txt};
		\addplot [dotted, mark=star]table[x=Ts,y=A2] {fig/availability.txt};
		\legend{$\switchsuccess = 0.99999$, $\switchsuccess = 0.999999$};
		\end{axis}
	\end{tikzpicture}
  \caption{Availability vs Rejuvenation Period}
  \label{fig:availabilityExp}
\end{figure} 

\subsection{Availability Maximization under Reliability Constraint}
\label{subsec:ava-max}

Assume the system longevity is $\lifetime$, if the rejuvenation
number is fixed,  then the availability of the main processing unit $\mainprocessor$
is also fixed. According to Lemma~\ref{lm:minReliability},
the system reliability achieves its minimal value when
$\lifetime \ \mathtt{mod} \ \rejuvenateperiod = 0$.
For the availability maximization problem, we use the worst
case reliability (Eq.~(\ref{eq:reliability})) to check
the reliability requirement.

From Eq.~(\ref{eq:availability}), it is easy to see that the
availability increases as rejuvenation period increases, i.e., the
number of rejuvenation times decreases. If the system does not perform
any rejuvenation, the availability is $100\%$. However, the system
also has the reliability requirement $\reliabilityrequire$.  According
to the analysis in Section~\ref{sec:reliability}, it is possible that
the system's reliability decreases below the required level
$\reliabilityrequire$ without rejuvenation if the operational time
$\lifetime$ is long enough.  Hence, there is an optimal rejuvenation
period that maximizes the availability and at the same time guarantees
the satisfaction of reliability requirement.  The MAX-AVA algorithm 
given in Algorithm~\ref{alg:maxava} is designed to find such an
optimal rejuvenation period.

In particular, the MAX-AVA algorithm initially assumes that the system
does not need to perform rejuvenation, i.e., assumes $n = 0$ and
$\rejuvenateperiod = L$ (line 1 -2). If the reliability requirement is
violated when $\rejuvenateperiod = \lifetime$, we need to sacrifice
availability by increasing the number of rejuvenation times until
$\reliabilityrequire$ is satisfied (line 3-7). If the number of
rejuvenations needed is too large that causes total rejuvenation time
exceed the required system longevity, the algorithm returns $-1$,
signaling that the system fails to achieve the reliability requirement
(line 8-10). Hence, system reboot becomes necessary.

Once we obtain the rejuvenation period ($\rejuvenateperiod$) from the
MAX-AVA algorithm, system maximal availability can be calculated by
Eq.~\eqref{eq:availability}. We assume $\availability = 0$ if the
reliability requirement can not be satisfied, i.e., if MAX-AVA returns
-1.

\begin{algorithm}
\caption{\textsc{MAX-AVA}}
\label{alg:maxava}
\begin{algorithmic}[1]
\REQUIRE System longevity $\lifetime$,
		 reliability constraint $\reliabilityrequire$.
\ENSURE The optimal rejuvenation period $\rejuvenateperiod$
	that maximize system availability.

\STATE $n = 0$ 
\STATE $\rejuvenateperiod = \lifetime / (n+1)$
\WHILE {$R(\lifetime, \lifetime/n) < \reliabilityrequire \land n \le
  \lifetime$} 
	\STATE // $R(\lifetime, \lifetime/n)$ is calculated according
        to Eq.~\eqref{eq:reliability} 
	\STATE $n = n+1$
	\STATE $\rejuvenateperiod = \lifetime / (n+1)$
\ENDWHILE
\IF {$n > \lifetime$}
	\STATE $\rejuvenateperiod = -1$ \qquad // indicating failure 
\ENDIF
\RETURN $\rejuvenateperiod$
\end{algorithmic}
\end{algorithm}

\subsection{Simulation Results}
We use simulation to reveal the relationship between rejuvenation
period and availability of the main processing unit under a given system
reliability constraint.  The simulation parameters are set as
following:

\begin{itemize}
\item Failure rate: $\failrate(t) = 3 t^2 / 10^9$
\item Probability of a successful task
migration between $\mainprocessor$ and $\backupprocessor$:
$\switchsuccess \in \{0.99999, 0.999999\}$
\item Reliability requirement: $\reliabilityrequire = 0.9997$
\item Longevity: $\lifetime = 100$
\item Rejuvenation time cost: $\rejuvenatetime = 0.5$
\item Rejuvenation periods: $\rejuvenateperiod \in \{1,5,10,\dots,95,100\}$
\end{itemize}

For each rejuvenation period, we use Eq.~\eqref{eq:availability}
to calculate availability of the main
processing unit. We assume $\availability = 0$ if the reliability
requirement can not be satisfied.
Fig.~\ref{fig:availabilityExp} shows the availability
under different rejuvenation periods. 
From Fig.~\ref{fig:availabilityExp}, we have the
following observations:
\begin{enumerate}
\item In general, when rejuvenation period increases, system's
  availability increases. 
\item System availability has a maximum value.  In particular, for a
  given system longevity value $\lifetime = 100$, the maximal 
  system availability is 99.5\% for both $\switchsuccess =
  0.99999$ and $\switchsuccess = 0.999999$. 
\item Too small or too large rejuvenation periods cause system
  reliability to decrease below the required level, hence causes the
  system to become unavailable, i.e., availability is $0$. In
  particular, the system is unavailable when $\rejuvenateperiod >60$
  for both cases, and when $\rejuvenateperiod < 10$ for
  $\switchsuccess = 0.99999$.
\end{enumerate}

By applying Algorithm~\ref{alg:maxava}, we obtain that when
$\rejuvenateperiod = 50$, the system achieves its maximal availability
of $99.5\%$, which is consistent with the simulation results depicted
in Fig.~\ref{fig:availabilityExp}.

\section{Delay Minimization}
\label{sec:scheduling}
\subsection{Scheduling Algorithm}
\label{subsec:scheduling}

When the main processing unit is performing rejuvenation process, some or
all of the tasks deployed on the main processing unit may have to be
migrated to the backup processing unit for their executions. In order to
guarantee that real-time tasks satisfy their deadlines, the non
real-time tasks deployed on the backup processing unit may have to be
postponed. To optimize the system's QoS , the delay of non real-time
tasks on the backup processing unit caused by the main processing unit going
through rejuvenation shall be minimized.  Clearly, if rejuvenation
takes place at the time when the main processing unit is idle, we can
utilize the idle time and hence reduce the delay of non real-time
tasks on the backup processing unit.

For real-time systems, priority-driven scheduling by definition never
intentionally leaves resources idle, i.e., a resource becomes idle
\textit{only} when there is no ready job in the waiting
queue~\cite{JaneLiuBook}.  The Earliest Deadline First (EDF)
scheduling algorithm is one of the most commonly used priority-driven
scheduling algorithms for real-time
systems~\cite{CLLiu1973}. Fig.~\ref{fig:idle}(a) shows an example of
task set $\taskset$'s schedule based on EDF scheduling algorithm,
where $\taskset = \{ \task_1 (3,1), \task_2 (4,1), \task_3 (6,1) \}$.
For priority-driven scheduling algorithms, we have the following
observation:

\begin{observation}
\label{obs:idleEDF}
For priority-driven scheduling algorithms, such as EDF, the longest
idle time interval often occurs towards the end of a task set's
hyper-period.  \finish
\end{observation}

The reverse of EDF is the Latest Release Time (LRT) scheduling
algorithm~\cite{JaneLiuBook} which schedules jobs backwards from the
latest deadline of all jobs to the earliest release time.  For the
same task set given above, Fig.~\ref{fig:idle}(b) gives the schedule
based on the LRT scheduling algorithm.  For the LRT scheduling
algorithm, we have the following observation:

\begin{observation}
\label{obs:idleLRT}
For the LRT scheduling algorithm, the longest idle time interval often
occurs towards the begin of a task set's hyper-period.  \finish
\end{observation}

Observation~\ref{obs:idleEDF} and Observation~\ref{obs:idleLRT} are
manifested in Fig.~\ref{fig:idle}, where the longest idle interval is
2 time units, which occur in the interval of $[10, 12]$ (end of
hyper-period), and $[0, 2]$ (beginning of hyper-period) with EDF and
LRT, respectively.

\begin{figure}[ht]
\centering
\includegraphics[width = 0.7\textwidth]{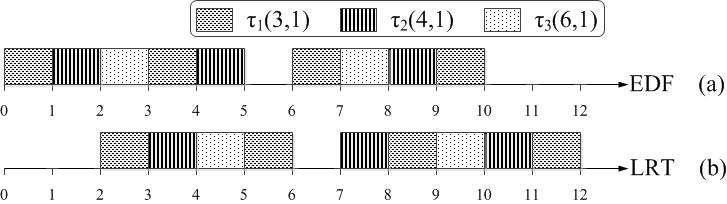}
\caption{EDF and LRT Scheduling for Task Set $\taskset$}
\label{fig:idle}
\end{figure}

The two observations provide us with the design base for our MIN-DELAY algorithm.
We use the following example to explain the intuitions.

\begin{example}
\label{ex:idleMotivation}
Consider the same periodic real-time task set $\taskset = \{ \task_1
(3,1), \task_2 (4,1), \task_3 (6,1) \}$ with hyper-period
$\hyperperiod = 12$, assume each rejuvenation takes $\rejuvenatetime =
2$ to complete and rejuvenation period is $\rejuvenateperiod = 7$,
then the first rejuvenation starts at time $t=7$.

If we use EDF to schedule the task set $\taskset$, the delay for non
real-time tasks on the backup processing unit is $\delay = 2$, and the
delay time interval is $[7, 9]$, as shown in
Fig.~\ref{fig:rejuvenation}.

\begin{figure}[ht]
\centering
\includegraphics[width = 0.7\textwidth]{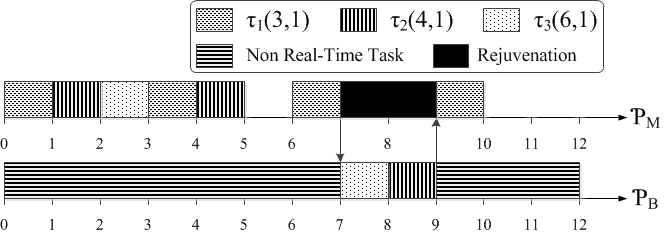}
\caption{Non Real-Time Task Delay}
\label{fig:rejuvenation}
\end{figure}

However, as shown in Fig.~\ref{fig:idle}, the EDF scheduling has an
idle time interval $[5,6]$. If we start the rejuvenation at time
$t=5$, we can utilize the idle time to reduce the delay.
Additionally, if we push the second idle time interval forward to the
rejuvenation starting time, we can further reduce the delay. Based on
Observation~\ref{obs:idleLRT}, we can use the LRT algorithm to
schedule jobs that are released after the rejuvenation starting time
to maximize the continuous idle time
interval.

Fig.~\ref{fig:idleMotivation} shows a schedule that not only
guarantees real-time tasks meeting their deadlines, but also allow
rejuvenation to take place without delaying any non real-time tasks.
In this case, the delay is $\delay = 0$ and the rejuvenation takes place in
time interval $[5, 7]$.

\begin{figure}[ht]
	\centering
	\includegraphics[width = 0.7\textwidth]{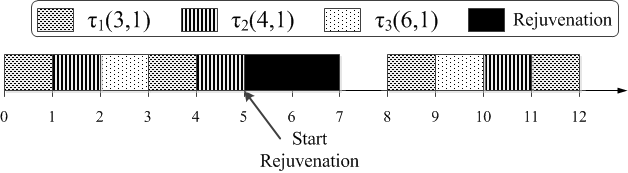} 
	\caption{Motivation Scheduling of $\taskset$}
	\label{fig:idleMotivation}
\end{figure}

\finish
\end{example}

To generalize the strategy used in Example~\ref{ex:idleMotivation},
assume a given optimal rejuvenation start time $t$ is within the task
set's $k$th hyper-period, i.e., $k\hyperperiod \le t \le
(k+1)\hyperperiod$, the rejuvenation process may take place in a time
interval in $[(k-1)\hyperperiod,(k+n+2)\hyperperiod]$, i.e., $[\effectivestart,
\effectivestart+\rejuvenatetime] \subseteq [(k-1)\hyperperiod,(k+n+2)\hyperperiod]$,
where $\effectivestart$ is the actual rejuvenation
start time which equals to the last idle time before $t$,
and $n=\lfloor \rejuvenatetime / \hyperperiod \rfloor$.  We
use EDF to schedule jobs released within $[(k-1)\hyperperiod,
\effectivestart]$, and use
LRT to schedule jobs that are released within
$[\effectivestart,(k+n+2)\hyperperiod]$ to push idle time towards the
beginning of the interval, i.e., towards $\effectivestart$.
Fig.~\ref{fig:schedule} shows the scheduling strategy.  As both EDF
and LRT are optimal from schedulability
perspective~\cite{JaneLiuBook}, hence, our scheduling strategy has the
same schedulability as EDF or LRT scheduling algorithm.

\begin{figure}[ht]
	\centering
	\includegraphics[width = 0.7\textwidth]{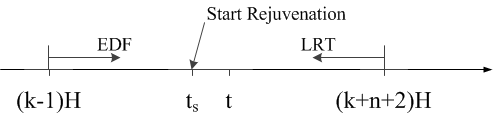} 
	\caption{Mixed Scheduling}
	\label{fig:schedule}  
\end{figure}

According to above analysis, the actual rejuvenation start time
$\effectivestart$ is smaller than the computed optimal rejuvenation
start time $t$, which lowers the rejuvenation period $\rejuvenateperiod$.
During the system longevity $\lifetime$, if the rejuvenation period
is lowered too much, the rejuvenation number may becomes larger than
the original rejuvenation number, which also enlarges the delay $\delay$
of non real-time applications on the backup processing unit $\backupprocessor$.
To minimize the delay $\delay$, the minimal rejuvenation period
$\rejuvenateperiodmin$ must guarantee that the rejuvenation number
with $\rejuvenateperiodmin$ is equal to the original rejuvenation number
with given rejuvenation period $\rejuvenateperiod$, i.e.,
\begin{align}
\label{eq:minRejuPeriod}
\rejuvenateperiodmin = \min \{ \rejuvenateperiodmin \in \mathbb{N} | \lfloor \lifetime / \rejuvenateperiodmin \rfloor = \lfloor \lifetime / \rejuvenateperiod \rfloor\}
\end{align}

In Example~\ref{ex:idleMotivation}, the rejuvenation starts at time 5
which is two time unit earlier than its scheduled rejuvenation time 7.
Based on the system reliability analysis in
Section~\ref{sec:reliability}, shorter rejuvenation period may cause
the system not meeting its reliability requirement
$\reliabilityrequire$.  Hence, we have to verify system reliability
requirement before changing the actual rejuvenation start time.

As discussed in Section~\ref{sec:reliability}, when the rejuvenation
period increases, the system reliability first increases and then
decreases. Hence, we can calculate the minimal rejuvenation period
$\rejuvenateperiodrequire$ that satisfies the system reliability
requirement $\reliabilityrequire$ using
Eq.~\eqref{eq:realReliability}. To maintain reliability requirement,
we must guarantee that the actual rejuvenation start time is no less
than $\rejuvenateperiodrequire$.
Therefore, the actual rejuvenation start time $\effectivestart$
must be no less than $\max \{ \rejuvenateperiodmin, \rejuvenateperiodrequire \}$
to minimize the delay and maintain system reliability requirement.

We now give the MIN-DELAY scheduling algorithm in
Algorithm~\ref{alg:mindelay}. Given system with longevity
$\lifetime$, reliability requirement $\reliabilityrequire$, and the
optimal rejuvenation period $\rejuvenateperiod$, first, we calculate
$\rejuvenateperiodrequire$ that guarantees $\reliabilityrequire$,
$\rejuvenateperiodmin$ that maintains rejuvenation number, and
$T$ which is the minimal rejuvenation period satisfying
system reliability and rejuvenation number requirements (Line 1-3).
In the scheduling process,
the algorithm determines the actual rejuvenation start times that
satisfies system reliability and rejuvenation number requirements (Line 8-10) and schedules jobs
based on EDF or LRT depending on the job release time with respect to
the rejuvenation start time (Line 11-14).  The complexity of the
algorithm is $O(L^2)$.

\begin{algorithm}
\caption{\textsc{MIN-DELAY}}
\label{alg:mindelay}
\begin{algorithmic}[1]
\REQUIRE Real-time periodic task set $\taskset$, rejuvenation 
		period $\rejuvenateperiod$, rejuvenation cost
		$\rejuvenatetime$, system longevity $\lifetime$, and
		system reliability requirement $\reliabilityrequire$

\STATE Calculate $\rejuvenateperiodrequire$ that guarantees $\reliabilityrequire$ using
        Eq.~\eqref{eq:realReliability}
\STATE Calculate $\rejuvenateperiodmin$ using Eq.~\eqref{eq:minRejuPeriod}
\STATE $T = \max \{ \rejuvenateperiodmin, \rejuvenateperiodrequire \}$
\STATE $n=\lfloor \rejuvenatetime / \hyperperiod \rfloor$ 
\STATE Initialize actual rejuvenation start time $\effectivestart =
\rejuvenateperiod$ 
\STATE $t_1 = 0$
\WHILE {$\lifetime > \effectivestart$}		
	\IF {EDF has idle time during $[\effectivestart-(\rejuvenateperiod-T),\effectivestart]$}
		\STATE $\effectivestart$ is the first idle begin time during $[\effectivestart-(\rejuvenateperiod-T),\effectivestart]$
	\ENDIF
	
	\STATE Schedule jobs released during $[t_1, \effectivestart]$ with EDF
	\STATE $t_2 = \lfloor \effectivestart/\hyperperiod \rfloor \cdot \hyperperiod$		
	\STATE $t_1 = t_2 + (n+2)\hyperperiod$	
	\STATE Schedule jobs released during $[\effectivestart, t_1]$ with LRT

	\STATE $\effectivestart = \effectivestart + \rejuvenatetime + \rejuvenateperiod$
\ENDWHILE

\STATE Schedule jobs released during $[t_1, \lifetime]$ with EDF
\end{algorithmic}
\end{algorithm}

\subsection{Simulation Results}

In this section, we evaluate the performance of the proposed MIN-DELAY
scheduling algorithm and compare it with EDF and LRT scheduling
algorithms~\cite{CLLiu1973,JaneLiuBook}. Our evaluation criteria is the delay of
non real-time tasks on the backup processing unit.

\subsubsection{Task Set Utilization Impact}\

This set of experiments evaluates the performance
of the proposed MIN-DELAY scheduling algorithm under
different task set utilizations. The experiment settings
are given below.
\begin{itemize}
	\item Number of tasks in a task set: $5$
	\item Task period range: $[10,20]$
	\item Task set utilizations: $U_\Gamma \in \{0.3,0.4,\dots,1.0\}$
	\item System longevity: $\lifetime = 10,000,000$
	\item Optimal rejuvenation period: $\rejuvenateperiod = 2,000,000$
	\item Minimal rejuvenation period: $\rejuvenateperiodrequire = 1,900,000$
	\item Rejuvenation time cost: $\rejuvenatetime = 100,000$
\end{itemize}

For each utilization option, we randomly generate 100 task sets
with the UUniform algorithm~\cite{Bini2005MetricGeneration}.
We schedule each task set and compute delays on the backup processing unit
with EDF, LRT, and MIN-DELAY algorithms, respectively. The average value
is used to represent the performance of each algorithm.

Fig.~\ref{fig:schedulingExp} shows the delay under different task set
utilizations. From Fig.~\ref{fig:schedulingExp}, we have the following
observations:
\begin{enumerate}
	\item For all scheduling algorithms, the delay increases when task
	set utilization increases.
	\item The proposed MIN-DELAY algorithm outperforms the EDF and LRT algorithms
	by as much as $9.01\%$ and $14.24\%$ under different task set utilizations, respectively.
	\item The performance advantage of MIN-DELAY algorithm decreases
	when task set utilization increases. In particular, the MIN-DELAY
	algorithm results in $9.01\%$ less delay than EDF when task set
	utilization is $0.3$, while the two algorithms have the same delay
	when task set utilization reaches $1.0$. 
\end{enumerate}

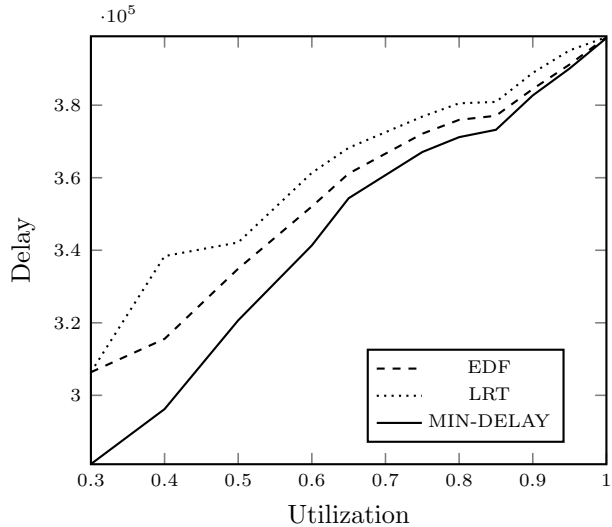
\begin{figure}[ht]
	\centering
	\pgfplotsset{
tick label style={font=\scriptsize},
label style={font=\scriptsize},
legend style={font=\scriptsize},
every axis/.append style={
thick,
tick style={semithick}}
}

\begin{tikzpicture}
	\begin{axis}[
		xlabel=\text{Utilization},
		xmin=0.3,xmax=1,
	    xtick={0.3,0.4,0.5,0.6,0.7,0.8,0.9, 1.0},
		xlabel near ticks,
		ymin=281000,ymax=399000,
		ylabel=\text{Delay},
		ylabel near ticks,
		legend style={at={(0.535,0.16)},anchor=west}
		]
	
	\addplot [dashed]table[x=U,y=EDFdelay] {fig/schedulingExp.txt};
	\addplot [dotted]table[x=U,y=LRTdelay] {fig/schedulingExp.txt};
	\addplot [solid]table[x=U,y=delay] {fig/schedulingExp.txt};
	\legend{EDF, LRT, MIN-DELAY}
	\end{axis}
\end{tikzpicture}
	\caption{Delay vs Utilization}
	\label{fig:schedulingExp}
\end{figure}

\subsubsection{Rejuvenation Time Cost Impact}\

The second set of experiments is to evaluate rejuvenation time
($\rejuvenatetime$) impact on the performance of the proposed
MIN-DELAY scheduling algorithm. The experiment settings are the same
as the previous experiments except that we fix the task set
utilization at 0.6 and set $\rejuvenatetime \in \{100,000,~ 110,000,\\ \dots,~ 200,000\}$.

Fig.~\ref{fig:rejuCostExp} shows the delay under different
rejuvenation time $\rejuvenatetime$. From Fig.~\ref{fig:rejuCostExp},
we have the following observations:
\begin{enumerate}
	\item For all scheduling algorithms, the delay increases when the
	rejuvenation time cost increases.
	\item The proposed MIN-DELAY algorithm outperforms the EDF and LRT algorithm
	by as much as $10.21\%$ and $10.23\%$ under different rejuvenation time costs, respectively.
	\item The EDF and LRT scheduling algorithms have similar performance. 
\end{enumerate}

\begin{figure}[ht]
	\centering
	\pgfplotsset{
tick label style={font=\scriptsize},
label style={font=\scriptsize},
legend style={font=\scriptsize},
every axis/.append style={
thick,
tick style={semithick}}
}

\begin{tikzpicture}
	\begin{axis}[
		xlabel=\text{Rejuvenation Time Cost},
		xmin=100000,xmax=200000,
		xlabel near ticks,
		ymin=329000,ymax=773000,
		ylabel=\text{Delay},
		ylabel near ticks,
		legend style={at={(0.535,0.16)},anchor=west}
		]
	
	\addplot [dashed]table[x=RC,y=EDFdelay] {fig/rejuCostExp.txt};
	\addplot [dotted]table[x=RC,y=LRTdelay] {fig/rejuCostExp.txt};
	\addplot [solid]table[x=RC,y=delay] {fig/rejuCostExp.txt};
	\legend{EDF, LRT, MIN-DELAY}
	\end{axis}
\end{tikzpicture}
	\caption{Delay vs Rejuvenation Time Cost ($U_\Gamma = 0.6$)}
	\label{fig:rejuCostExp}
\end{figure}
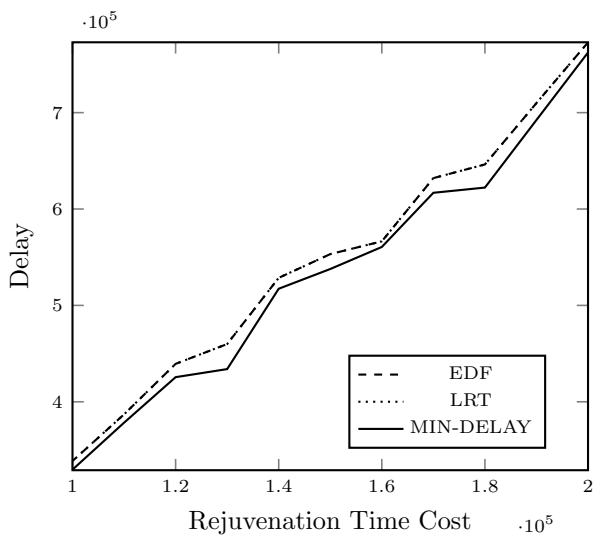

Both sets of experiments show that the proposed MIN-DELAY algorithm
has advantages over the EDF algorithm with respect to application
execution delay on the backup processing unit.

\section{Conclusion}
\label{sec:conclusion}
In this paper, we use software rejuvenation as a preventive technique
to improve system's QoS for long-running applications with real-time constraints.
We have formally analyzed the relationship
between software rejuvenation frequency and system reliability, longevity,
and availability. Based on the theoretic analysis, we have
developed approaches to maximizing system reliability, longevity, and availability,
and minimizing application execution delays on the backup processing unit. The developed
semi-priority-driven scheduling algorithm, i.e., the MIN-DELAY
scheduling algorithm can reduce application delay by
$9.01\%$ and $14.24\%$ over the EDF and LRT scheduling algorithms, respectively.

\section*{Acknowledgment}
\label{sec:acknowledgement}
The research is supported in part by NSF CNS 1545008.

\bibliographystyle{abbrv}
\bibliography{ref}

\end{document}